\documentclass[11pt]{amsart}
\usepackage{amsmath,amssymb}
\usepackage{graphicx}

\begin{document}

\newtheorem{thm}{Theorem}[section]
\newtheorem{lem}[thm]{Lemma}
\newtheorem{prop}[thm]{Proposition}
\newtheorem{cor}[thm]{Corollary}
\newtheorem{defn}[thm]{Definition}
\newtheorem*{remark}{Remark}

\numberwithin{equation}{section}

\newcommand{\Z}{{\mathbb Z}} 
\newcommand{\Q}{{\mathbb Q}}
\newcommand{\R}{{\mathbb R}}
\newcommand{\C}{{\mathbb C}}
\newcommand{\N}{{\mathbb N}}
\newcommand{\FF}{{\mathbb F}}
\newcommand{\T}{{\mathbb T}}
\newcommand{\fq}{\mathbb{F}_q}

\def\scrA{{\mathcal A}}
\def\scrB{{\mathcal B}}
\def\scrD{{\mathcal D}}
\def\scrH{{\mathcal H}}
\def\scrL{{\mathcal L}}
\def\scrM{{\mathcal M}}
\def\scrN{{\mathcal N}}
\def\scrS{{\mathcal S}}

\newcommand{\rmk}[1]{\footnote{{\bf Comment:} #1}}

\renewcommand{\mod}{\;\operatorname{mod}}
\newcommand{\ord}{\operatorname{ord}}
\newcommand{\TT}{\mathbb{T}}
\renewcommand{\i}{{\mathrm{i}}}
\renewcommand{\d}{{\mathrm{d}}}
\renewcommand{\^}{\widehat}
\newcommand{\HH}{\mathbb H}
\newcommand{\Vol}{\operatorname{vol}}
\newcommand{\area}{\operatorname{area}}
\newcommand{\tr}{\operatorname{tr}}
\newcommand{\norm}{\mathcal N} 
\newcommand{\intinf}{\int_{-\infty}^\infty}
\newcommand{\ave}[1]{\left\langle#1\right\rangle} 
\newcommand{\Var}{\operatorname{Var}}
\newcommand{\Cov}{\operatorname{Cov}}
\newcommand{\Prob}{\operatorname{Prob}}
\newcommand{\sym}{\operatorname{Sym}}
\newcommand{\disc}{\operatorname{disc}}
\newcommand{\CA}{{\mathcal C}_A}
\newcommand{\cond}{\operatorname{cond}} 
\newcommand{\lcm}{\operatorname{lcm}}
\newcommand{\Kl}{\operatorname{Kl}} 
\newcommand{\leg}[2]{\left( \frac{#1}{#2} \right)}  
\newcommand{\SL}{\operatorname{SL}}
\newcommand{\PSL}{\operatorname{PSL}}

\newcommand{\sumstar}{\sideset \and^{*} \to \sum}

\newcommand{\LL}{\mathcal L} 
\newcommand{\sumf}{\sum^\flat}
\newcommand{\Hgev}{\mathcal H_{2g+2,q}}
\newcommand{\USp}{\operatorname{USp}}
\newcommand{\conv}{*}
\newcommand{\dist} {\operatorname{dist}}
\newcommand{\CF}{c_0} 
\newcommand{\kerp}{\mathcal K}

\newcommand{\gp}{\operatorname{gp}}
\newcommand{\Area}{\operatorname{Area}}
\newcommand{\Isom}{\operatorname{Isom}}

\title[eigenvalue spacing distribution for a point scatterer]{On the eigenvalue spacing distribution for a point scatterer on the flat torus}
\author{Ze\'ev Rudnick \and Henrik Uebersch\"ar}
\address{The Raymond and Beverly Sackler Schoool of Mathematical Sciences, Tel Aviv University, Tel Aviv 69978,
Israel.}
\email{rudnick@post.tau.ac.il}
\address{The Raymond and Beverly Sackler Schoool of Mathematical Sciences, Tel Aviv University, Tel Aviv 69978,
Israel.}
\email{henrik@post.tau.ac.il}
\date{\today}
\thanks{Z.R. was
partially supported by the Israel Science Foundation (grant No.
1083/10). H.U. was supported by a Minerva Fellowship.}

\begin{abstract}
 We study the level spacing distribution for the spectrum of a point scatterer on a flat torus. In the $2$-dimensional case, we show that in the weak coupling regime the eigenvalue spacing distribution coincides with that of the spectrum of the  Laplacian (ignoring multiplicties), by showing that the perturbed eigenvalues generically clump with the unperturbed ones on the scale of the mean level spacing.
We also study the three dimensional case, where the situation is very different.
\end{abstract}

\maketitle
\section{Introduction}

\subsection{The  \u{S}eba billiard}
Point scatterers are toy models used to understand aspects of
quantum systems for which the corresponding classical limit is
intermediate between integrable and chaotic. In this paper we study
the spectral statistics of point scatterers on the flat torus (a
"\u{S}eba billiard") in the "weak coupling" regime.

A point scatterer on the torus is formally given by a Hamiltonian
\begin{equation}\label{point scatterer}
-\Delta+\alpha\delta_{x_0}, \quad \alpha\in\R
\end{equation}
where $\Delta$ is the Laplacian, $\alpha$ denotes a coupling
constant and $x_0$ denotes the position of the scatterer.
Mathematically a point scatterer is realised as a self-adjoint
extension of the Laplacian $-\Delta$ acting on functions which
vanish near $x_0$ (see \cite{CdV}). Such extensions are
parameterized by a phase $\varphi\in(-\pi,\pi]$, where $\varphi=\pi$
corresponds to the standard Laplacian ($\alpha=0$ in \eqref{point
scatterer}). We denote the corresponding operator by
$-\Delta_{x_0,\varphi}$.
For $\varphi=\pi$ the eigenvalues are those of the standard
Laplacian. For $\varphi\neq\pi$ ($\alpha\neq 0$) the resulting
spectral problem still has the eigenvalues of the unperturbed
problem, with multiplicity decreased by one, as well as a new set of
$\Lambda_\varphi=\lbrace\lambda_j^\varphi\rbrace$ of eigenvalues
interlaced between the sequence of unperturbed eigenvalues, each
appearing with multiplicity one, and satisfying the spectral
equation
\begin{equation}\label{speceqn}
\sum_{j=0}^\infty
|\psi_j(x_0)|^2\left(\frac{1}{\lambda_j-\lambda}-\frac{\lambda_j}{\lambda_j^2+1}\right)=c_0\tan\frac{\varphi}{2}
\end{equation}
where
$c_0=\sum_{j=0}^\infty\frac{|\psi_j(x_0)|^2}{\lambda_j^2+1}$
and $\lbrace\psi_j\rbrace$ form an orthonormal basis of
eigenfunctions for the unperturbed problem:
$-\Delta\psi_j=\lambda_j\psi_j$. The eigenfunction corresponding to
$\lambda\in\Lambda_\varphi$ is the Green's function
$G_{\lambda}(x,x_0)=(\Delta+\lambda)^{-1}\delta_{x_0}$.

 We denote the unperturbed eigenvalues without multiplicities by
$$\scrN=\lbrace n_0=0<n_1<\dots <n_j<\dots \rbrace
$$
and call them the "norms" of the torus. Note that the perturbed
eigenvalues defined by \eqref{speceqn} are independent of the
location $x_0$ of the scatterer, since in the case of the torus the
sums $\sum_{\lambda_j=n}| \psi_j(x_0)|^2 = \#\{\lambda_j=n\}$ are
independent of $x_0$.

 \subsection{Spacing distributions}

The perturbed eigenvalues $\lbrace\lambda_j^\varphi\rbrace$
interlace with the norms  $\lbrace n_j\rbrace$ as
follows
\begin{equation}
\lambda_0^\varphi<0=n_0<\lambda_1^\varphi<n_1<\cdots<\lambda_k^\varphi<n_k<\cdots
\end{equation}
The nearest neighbour spacings for the norms and for the perturbed 
eigenvalues are defined by
\begin{equation}
\delta_j:=n_{j+1}-n_j, \qquad
\delta_j^\varphi:=\lambda_{j+1}^\varphi - \lambda_j^\varphi
\end{equation}
The {\em mean spacing} between the norms is  defined by
\begin{equation}
\ave{\delta_j}_x:=\frac 1{N(x)} \sum_{n_j\leq x} \delta_j \sim
\frac{x}{N(x)}, \quad x\to\infty
\end{equation}
where
\begin{equation}
N(x):=\#\{j:n_j\leq x \}
\end{equation}
and likewise for the mean spacing $\ave{\delta_j^\varphi}_x$ between
the new eigenvalues. Clearly
\begin{equation}
\ave{\delta_j^\varphi}_x \sim \ave{\delta_j}_x,\quad x\to \infty
\end{equation}

We define normalized nearest neighbour spacings by
\begin{equation}
\^\delta_j:=\frac{\delta_j}{\ave{\delta_j}_x},\qquad
\^\delta_j^\varphi:=\frac{\delta_j^\varphi}{\ave{\delta_j^\varphi}_x}
\end{equation}

We want to determine the distribution of the normalized spacings
$\delta_j^\varphi$.

Shigehara, Cheon et al. \cite{Shigehara1,Shigehara2,Shigehara3}
identify two regimes in the semiclassical limit for a point
scatterer in dimension $2$: In the weak coupling regime the phase
$\varphi$ is fixed as $\lambda\to\infty$. In this regime the authors
predict a Poissonian level spacing distribution for the perturbed
spectrum. The strong coupling regime is when $\varphi$ varies as
$\lambda\to\infty$ so as to satisfy:
$c_0\tan\frac{\varphi}{2}\sim-\frac{1}{4\pi}\log \lambda$.
where they predict level repulsion. In most numerical studies of
this problem, it is the second regime that appears, due to a
truncation procedure \cite{Seba}. For an analytic study of this
regime see \cite{Sieber, BogomolnyGerlandSchmit, RF}.

We deal  with the spectrum of a point scatterer in the weak coupling
regime ($\varphi$ fixed). We will show that the  level spacing
distributions of the norms and of the perturbed spectrum (if either
exists) coincide.
 Since it is generally believed that the spacing distribution of the norms  is Poissonian
(if the torus is either rational, such as the standard torus
$\R^2/\Z^2$, or generic irrational in a suitable sense \cite{EMM}), that would
imply that the perturbed spectrum is also Poissonian.


\subsection{Our results}
We denote the differences between the old and new eigenvalues by
\begin{equation}
d_j:=n_j-\lambda_j^\varphi > 0
\end{equation}
Since $\delta_j-\delta_j^\varphi = d_{j+1}-d_j$, the normalized
nearest neighbour spacings between the norms and the perturbed spectrum are
related by
\begin{equation}
\^\delta_j-\^\delta_j^\varphi \sim
\frac{d_{j+1}-d_j}{\ave{\delta_j}_x}
\end{equation}
We define the mean difference of $d_j$ by
\begin{equation}
\ave{d_j}_x=\frac{1}{N(x)}\sum_{\lambda_j^\varphi\leq x}d_j.
\end{equation}
We will show that the ratio between the mean difference $d_j$ and
the mean spacing $\delta_j$ vanishes:
\begin{thm}\label{thm main}
For a point scatterer on a  $2$-dimensional flat torus,
\begin{equation}
\frac{\ave{d_j}_x}{\ave{\delta_j}_x}\to0, \quad \mbox{as } x\to
\infty.
\end{equation}
\end{thm}

As a consequence, since the differences $d_j\geq 0$ are
non-negative, we deduce
\begin{cor}\label{cor:clumping}
Outside of a zero-density subsequence,
\begin{equation} \label{clumping}
\frac{d_j}{\ave{\delta_j}}\to 0 ,\quad \mbox{as }j \to \infty \;.
\end{equation}
\end{cor}
That is the norms and the perturbed eigenvalues clump together
generically\footnote{Recently Tudorovskiy, Kuhl and St\"ockmann
\cite{Stoeckmann}  presented a heuristic argument that in the fixed
regime the spacing distribution
should be Poissonian by claiming the bound \eqref{clumping} holds
individually, for {\em all} $j$. We are unable to verify this. } on
the scale of the mean spacing. Therefore:
\begin{cor}
If the spacings $\delta_j$ for the norms have a limiting distribution, then so do the spacings $\delta_j^\varphi$ for the 
perturbed spectrum and the limiting distributions coincide.
\end{cor}

A similar result holds for hyperbolic surfaces if the point
scatterer is placed in a generic position. We will not give the
details here.

\subsection{Dimension $3$}
The situation is very different for a three-dimensional torus
$\T^3$.
Let $\eta_j^\varphi$ be the perturbed eigenvalues of the point
scatterer and $\eta_j$ the unperturbed eigenvalues counted without
multiplicity (the norms). The ordering is
\begin{equation}
\eta_0^\varphi<0=\eta_0<\eta_1^\varphi <\eta_1<\dots
<\eta_j^\varphi<\eta_j
\end{equation}

As before we let $d_j:=\eta_j-\eta_j^\varphi$ and
$\delta_j=\eta_{j+1}-\eta_j$. We denote by $\ave{d_j}_x$ the average
of the spacings $d_j$, and by $\ave{\delta_j}_x$ of the spacings of
the norms, for $\eta_j^\varphi\leq x$.

\begin{thm}\label{intermediate}
For the $3$-dimensional flat torus, we have
\begin{equation}
\lim_{x\to\infty} \frac{\ave{d_j}_x}{\ave{\delta_j}_x}= \frac{1}{2}.
\end{equation}
\end{thm}
Note that Theorem~\ref{intermediate} does not give any information
on the relation between level spacing distributions for the norms and for the 
perturbed spectrum. For an empirical study of the spectral
statistics in dimension $3$, see \cite{ShigeharaCheon3D}.

\section{Overview of the proof}
\subsection{Our method}

We derive Theorems~\ref{thm main} and \ref{intermediate} from the
asymptotics as $\beta\to 0$ of the sum
$$\sum_{j=0}^{\infty}d_j e^{-\beta\lambda_j^\varphi}$$
To so, we approximate the sum by the difference
$$\sum_{j=0}^\infty\lbrace e^{-\beta\lambda_j^\varphi}-e^{-\beta n_j} \rbrace
$$
 of the heat traces of the operators $-\Delta_{x_0,\varphi}$ and $-\Delta$, which we study
 in \S~\ref{sec:using trace} by using a trace formula which will be developed in \S~\ref{Sec:Trace}, \ref{Sec:Trace2} in the $2$-dim case.
The $3$-dimensional case is treated in \S~\ref{sec:3d}.

\subsection{A trace formula for the point scatterer on the torus}
We work with a rectangular  two-dimensional flat torus
$\T^2=\R^2/2\pi\scrL_0$, where $\scrL_0=\Z(1/a,0)\oplus\Z(0,a)$ for
some $a>0$. Denote by $\scrL$ the dual lattice of $\scrL_0$. The
eigenvalues of the Laplacian on $\T^2$ are the norms of the vectors
of the dual lattice $\scrL$ (cf. section 2 in \cite{RU}). We denote
the set of norms of the dual lattice vectors by
\begin{equation}
\scrN=\lbrace0<n_1<\cdots\rbrace
\end{equation}
and the multiplicity of an eigenvalue $n\in\scrN$ is denoted by
\begin{equation}
r_\scrL(n)=\#\lbrace\xi\in\scrL:|\xi|^2=n\rbrace.
\end{equation}

Recall that the perturbed eigenvalues
$\lbrace\lambda_j^\varphi\rbrace$ interlace with the norms $\lbrace
n_j\rbrace$. The ordering is
\begin{equation}
\lambda_0^\varphi<0=n_0<\lambda_1^\varphi <n_1<\dots
<\lambda_j^\varphi<n_j.
\end{equation}
(That $\lambda_0^\varphi<0$ is given in  \cite{Hillairet}).

We denote $n_j=\rho_j^2$, where $\rho_j>0$ for $j\geq 1$, and
$\lambda_j^\varphi=(\rho_j^\varphi)^2$, where $\rho_j^\varphi>0$ if
$j\geq1$ and  $\Im\rho_0^\varphi>0$ (note that $\lambda_0^\varphi<0$
and $\rho_0^\varphi$ is pure imaginary). The spectral function
\begin{equation}\label{def spectral}
S_\varphi(\rho)=-\frac{1}{\rho^2}+\sum_{j=1}^\infty
r(n_j)\left\{\frac{1}{n_j-\rho^2}-\frac{n_j}{n_j^2+1}\right\}-c_0\tan\frac{\varphi}{2}
\end{equation}
has simple poles at the points $\rho=\pm\rho_j$ and zeroes at the
points $\rho=\pm\rho_j^\varphi$. For $\sigma$ large enough and
$\Im\rho=-\sigma$ we will show that
\begin{equation}
S_\varphi(\rho)=-\frac{1}{2\pi}\log(\i\rho)+\frac{1}{2\pi}D(\rho)+c(\varphi)
\end{equation}
for $c(\varphi)=c_1-c_0\tan\frac{\varphi}{2}$, where $c_1$ is some
real constant, and $|D(\rho)|\ll_\sigma 1$.

Let $h$ be an even function which is analytic in a strip
$|\Im\rho|\leq\sigma'$ for some $\sigma'>\sigma$ and satisfies
$$|h(\rho)|\ll(1+|\Re\rho|)^{-5-\delta}$$
for some $\delta>0$ uniformly in the same strip. We have the
following general trace formula which we prove in
sections~\ref{Sec:Trace}, \ref{Sec:Trace2}.
Let $\sigma>\sigma_0(\varphi)$ be sufficiently large.  Then  for all
$h$ as above, we have
\begin{equation} \label{pretraceA}
\begin{split}
\sum_{j=0}^\infty \lbrace h(\rho_j^\varphi)-h(\rho_j)\rbrace
&=\frac{1}{2\pi\i}\int_{-\i\sigma-\infty}^{-\i\sigma+\infty}\frac{h(\rho)d\rho}{\rho(\log\i\rho-2\pi c(\varphi))}\\
&-\frac{1}{2\pi\i}\int_{-\i\sigma-\infty}^{-\i\sigma+\infty}h'(\rho)\log\left(1-\frac{\scrD(\rho)}{\log\i\rho-2\pi
c(\varphi)}\right)d\rho.
\end{split}
\end{equation}

\subsection{A Tauberian theorem}

To prove Theorem~\ref{thm main} we will employ a Tauberian Theorem
and reduce the problem to studying the asymptotics as $\beta\searrow
0$ of
\begin{equation}
\tilde A(\beta) =\sum_j d_j e^{-\beta \lambda_j^\varphi}
\end{equation}
To study $\tilde A(\beta)$ we prove the following approximation (cf.
\eqref{approx} and Lemma \ref{approxlemma})
\begin{equation}
\sum_{j=0}^\infty d_j e^{-\beta\lambda_j^\varphi}=
\frac{1}{\beta}\sum_{j=0}^\infty \lbrace
e^{-\beta\lambda_j^\varphi}-e^{-\beta n_j}\rbrace +O(\beta^{-1/2}).
\end{equation}
We then use the trace formula \eqref{pretraceA} with
$h(\rho)=e^{-\beta\rho^2}$ to bound $\sum_{j=0}^\infty \lbrace
e^{-\beta\lambda_j^\varphi}-e^{-\beta n_j}\rbrace$ and obtain the
following estimate which is the key result in the proof of
Theorem~\ref{thm main}:
\begin{prop}\label{main bound}
As $\beta \searrow 0$,
\begin{equation}
\tilde A(\beta) =\sum_j d_j e^{-\beta \lambda_j^\varphi} \ll \frac
{1}{\beta \log \frac {1}{\beta}}.
\end{equation}
\end{prop}

\subsection{Proof of Theorem~\ref{thm main}}

we will use Karamata's Tauberian Theorem (see e.g. \cite{Feller}) 
which deals with the following situation: We say a positive function
$L(x)$ is {\em slowly varying} if $L(kt)\sim L(t)$ as $t\to \infty$
for each fixed $k>0$. We are given a non-decreasing function $A(t)$
on $\R_+$ such that the Laplace transform
\begin{equation}
\tilde A(\beta) :=\int_0^\infty e^{-\beta t} dA(t)
\end{equation}
converges for all $\beta>0$. Suppose there exists two real numbers
$c\geq 0$, $\omega>0$ and a slowly varying function $L(x)$ so that
\begin{equation}
\tilde A(\beta) = \left\{ c+o(1) \right\} \beta^{-\omega}
L(1/\beta), \quad \beta \searrow 0
\end{equation}
Then
\begin{equation}
 A(x) = \left\{ c+o(1) \right \} \frac{x^\omega L(x)} {\Gamma(\omega +1)}, \quad x\to \infty
\end{equation}

We apply Karamata's Tauberian theorem to the function
\begin{equation}
A(x):=\sum_{\lambda_j \leq x} d_j
\end{equation}
which is non-decreasing since $d_j\geq 0$. The Laplace transform
$\tilde A$ is
\begin{equation}
\tilde A(\beta)=\sum_j d_j e^{-\beta \lambda_j^\varphi} \;.
\end{equation}
Proposition~\ref{main bound} implies that $\tilde A(\beta) =
o(1/(\beta\sqrt{\log \frac{1}{\beta}}))$. Thus in Karamata's
theorem, we may take $\omega=1$,  $L(t) = 1/\sqrt{\log t}$, and
$c=0$ to find
\begin{equation}
A(x) = o\left(\frac x{\sqrt{\log x}}\right), \quad x\to \infty.
\end{equation}
Therefore
\begin{equation}
\frac{\ave{d_j}_x}{\ave{\delta_j}_x}=\frac{A(x)}{N(x)}\frac{N(x)}{x}=\frac{A(x)}{x}=o\left(\frac{1}{\sqrt{\log
x}}\right)
\end{equation}
as $x\to\infty$, proving Theorem~\ref{thm main}.


\subsection{Three-dimensional tori}
As in the $2$-dimensional case, Theorem~\ref{intermediate} follows
from the following proposition which we prove in
section~\ref{sec:3d}.
\begin{prop}
We have as $\beta\searrow0$
\begin{equation}
\sum_{j=0}^{\infty}d_j e^{-\beta\eta_j^\varphi} =
\frac{1}{2\beta}+O(\beta^{-3/4}).
\end{equation}
\end{prop}
The key tools in the derivation are a trace formula (cf. Theorem
\ref{TF3d}) and an approximation lemma (cf. Lemma \ref{approx3d}).

\section{The trace formula}\label{Sec:Trace}

We follow the same path as in \cite{U} for a compact quotient $\Gamma\backslash\HH$.
\subsection{Overview of the proof}
Let $T>0$ be such that $T\notin\lbrace\rho_j\rbrace\cup\lbrace\rho_j^\varphi\rbrace$, $\sigma>\Im\rho_0^\varphi$
and consider the box $$B(\sigma,T)=\lbrace \rho \mid |\Im\rho|\leq\sigma, |\Re\rho|\leq T\rbrace.$$

For $\sigma$ large enough and $\Im\rho=-\sigma$ we will show that
the spectral function \eqref{def spectral} can be written as
\begin{equation}\label{geomrepresentation}
S_\varphi(\rho)=-\frac{1}{2\pi}\log(\i\rho)+\frac{1}{2\pi}D(\rho)+c(\varphi)
\end{equation}
for $c(\varphi)=c_1-c_0\tan\frac{\varphi}{2}$, where $c_1$ is some real constant, and $|D(\rho)|\ll_\sigma 1$.

Let $h$ be an even function which is analytic in a strip
$|\Im\rho|\leq\sigma'$ for some $\sigma'>\sigma$ and satisfies
$$|h(\rho)|\ll(1+|\Re\rho|)^{-5-\delta}$$ for some $\delta>0$
uniformly in the same strip. A contour integration gives
\begin{equation}\label{boxcontint}
2\sum_{\rho_j^\varphi\in B(\sigma,T)}h(\rho_j^\varphi)-2\sum_{\rho_j\in B(\sigma,T)}h(\rho_j)
=\frac{1}{2\pi\i}\int_{\partial B(\sigma,T)^+}h(\rho)\frac{S_\varphi'}{S_\varphi}(\rho)d\rho
\end{equation}
We may rewrite \eqref{boxcontint} as
\begin{equation}
\begin{split}
&2h(\rho_0^\varphi)-2h(0)+2\sum_{0<\rho_j^\varphi<T}h(\rho_j^\varphi)-2\sum_{0<\rho_j<T}h(\rho_j)\\
=&\frac{1}{\pi\i}\int_{-\i\sigma-T}^{-\i\sigma+T}\frac{h(\rho)d\rho}{\rho(\log\i\rho -2\pi c(\varphi))}\\
&-\frac{1}{\pi\i}\int_{-\i\sigma-T}^{-\i\sigma+T}h'(\rho)\log\left(1-\frac{\scrD(\rho)}{\log\i\rho-2\pi
c(\varphi)}\right)d\rho +\partial B(T)
\end{split}
\end{equation}
where
\begin{equation}
\begin{split}
\partial B(T)=&\frac{1}{\pi\i}\left[h(r)\log\left(\frac{S_\varphi(r)}{\log(\i r)-2\pi c(\varphi)}\right)\right]_{-\i\sigma-T}^{-\i\sigma+T}\\
&+\frac{1}{\pi\i}\int_{-\i\sigma+T}^{\i\sigma+T}h(\rho)\frac{S'_{\varphi}}{S_\varphi}(\rho)d\rho.
\end{split}
\end{equation}

Choose a sequence $\lbrace T_n\rbrace$ away from
$\lbrace\rho_j\rbrace\cup\lbrace\rho_j^\varphi\rbrace$ such that
$\lim_n T_n=\infty$. By use of the asymptotics
\eqref{geomrepresentation} we show that the integral over the
contour $[-\i\sigma-T_n,-\i\sigma+T_n]$ converges absolutely as
$n\to\infty$. Since Weyl's law implies that both traces converge
absolutely, it follows that $\lim_n \partial B(T_n)$ exists. The
main step in the proof of the trace formula is to show that actually
$\lim_n \partial B(T_n)=0$ for a suitable choice of a sequence
$\lbrace T_n\rbrace$.

In Lemma \ref{polybound} we construct a sequence $\lbrace
T_n\rbrace$ which satisfies
$$
|S_\varphi(T_n+\i w)|\ll_\epsilon T_n^{4+\epsilon}.
$$
We then use this bound together with our knowledge of the existence
of   $\lim_n \partial B(T_n)$, which holds in particular for a
certain test function $h_5$ with suitable symmetry properties (cf.
Lemma \ref{testfunction}). We exploit the properties of this
particular test function to bound $\log|S_\varphi|$ on average on
the segments $[-\i\sigma+T_n,T_n]$, namely (cf. Lemma \ref{avbound})
\begin{equation}\label{avbound2}
\left|\int_{-\i\sigma+T_n}^{T_n}\log|S_\varphi(\rho)|d\rho\right|\ll T_n^5.
\end{equation}
which allows us to pass to the limit $T_n\to \infty$ and obtain the
trace formula~\eqref{pretraceA}.
\begin{remark}
We are unable to obtain an individual bound on $\log|S_\varphi(\rho)|$ on the segments $[-\i\sigma+T_n,T_n]$,
but a bound on average suffices for our purposes.
\end{remark}

\subsection{The Green's function on the torus}
The free Green's function on $\R^2$ is given by
\begin{equation}
g_\lambda(x,x_0)=\frac{1}{2\pi}K_0(\i\rho |x-x_0|),\qquad \lambda=\rho^2,
\end{equation}
where $K_0$ denotes the zeroth Bessel function.

From the integral representation 
\begin{equation}
K_0(r)=\int_0^\infty \exp(-r\cosh t)dt
\stackrel{w=\cosh t}{=}\int_1^\infty \frac{e^{-wr}dw}{\sqrt{w^2-1}}, \quad \Re r>0
\end{equation}
we obtain the following integral representation for the free Green's function on $\R^2$
\begin{equation}
g_\lambda(x,x_0)=\frac{1}{2\pi}\int_1^\infty\frac{e^{-\i w\rho|x-x_0|}dw}{\sqrt{w^2-1}},\qquad \lambda=\rho^2,\qquad \Im\rho<0.
\end{equation}
We derive an integral representation for the Green's function on the torus $\T^2$ by the method of images. Let $\Im\rho<0$. We have
\begin{equation}
G_\lambda(x,x_0)
=\sum_{n\in\scrL}g_\lambda(x+n,x_0)
=\frac{1}{2\pi}\int_1^\infty\frac{k_\rho(w;x,x_0)dw}{\sqrt{w^2-1}}
\end{equation}
where
\begin{equation}
k_\rho(w;x,x_0)=\sum_{n\in\scrL}e^{-\i w\rho|x-x_0+n|}.
\end{equation}
Absolute convergence follows from the inequality (note $w\geq1$)
\begin{equation}
|k_\rho(w;x,x_0)|\leq\sum_{n\in\scrL}e^{-\sigma w|x-x_0+n|}
\ll 1+\sum_{0\neq m\in\scrN} r_\scrL(m)e^{-\sigma\sqrt{m}}
\end{equation}

The Bessel function has the asymptotics
\begin{equation}
K_0(z)=-\log(z/2)-\gamma+o(1), \quad z\to0
\end{equation}
where $\gamma$ denotes Euler's constant.
Therefore the free Green's function has the asymptotics
\begin{equation}
g_\lambda(x,x_0)=-\frac{1}{2\pi}\log(\i\rho|x-x_0|/2)-\frac{\gamma}{2\pi}+o(1)
\end{equation}
as $x\to x_0$.
Thus we have the following asymptotics for the Green's function on the torus
\begin{equation}
G_\lambda(x,x_0)=-\frac{1}{2\pi}\log(\i\rho|x-x_0|/2)-\frac{\gamma}{2\pi}+C_\lambda+o(1)
\end{equation}
as $x\to x_0$, where $$C_\lambda=\sum_{n\in\scrL\setminus\lbrace0\rbrace}g_\lambda(x_0+n,x_0).$$

\subsection{}
In view of the spectral expansion of the Green's function $G_\lambda$ the spectral function may be written as
\begin{equation}\label{specfunction}
S_\varphi(\rho)=\lim_{x\to x_0}\lbrace G_\lambda(x,x_0)-\Re G_\i(x,x_0) \rbrace-c_0\tan\frac{\varphi}{2}
\end{equation}
where $$c_0= 1+\sum_{0\neq   n\in\scrN}\frac{r_\scrL(n)}{n^2+1}.$$

We may rewrite \eqref{specfunction} as
\begin{equation}\label{specid}
S_\varphi(\rho)=-\frac{1}{2\pi}\log\i\rho+\frac{1}{2\pi}\scrD(\rho)+c(\varphi)
\end{equation}
where
\begin{equation}
k(x)=\sum_{n\in\scrL\setminus\lbrace0\rbrace}e^{-\i x|n|}
=\sum_{0\neq m\in\scrN} r_\scrL(m)e^{-\i x\sqrt{m}}
\end{equation}
\begin{equation}
\scrD(\rho)=\int_1^\infty \frac{k(\rho w)dw}{\sqrt{w^2-1}}.
\end{equation}
and
\begin{equation}
c(\varphi)=-\frac{1}{2\pi}\Re\scrD(-e^{\i\pi/4})-c_0\tan\frac{\varphi}{2}
\end{equation}
is a real constant.

We have the expression
\begin{equation}
c(\varphi)=c_1-c_0\tan\frac{\varphi}{2}
\end{equation}
where
\begin{equation}
c_1=-\frac{1}{2\pi}\sum_{m\in\scrN} r_\scrL(m)\int_1^\infty \frac{\cos(\sqrt{\frac{m}{2}}w)e^{-\sqrt{\frac{m}{2}}w}dw}{\sqrt{w^2-1}}.
\end{equation}

\begin{lem}
For sufficiently large $\sigma>0$ and $\Im\rho=-\sigma$
\begin{equation}\label{cond}
\frac{|\scrD(\rho)|}{|\log\i\rho-2\pi c(\varphi)|}<1.
\end{equation}
\end{lem}
\begin{proof}
We have
\begin{equation}
|\log\i\rho-2\pi c(\varphi)|\geq|\log\sqrt{\sigma^2+(\Re\rho)^2}-2\pi c(\varphi)|
\geq\log\sigma-2\pi |c(\varphi)|
\end{equation}
and
\begin{equation}\label{o1bound}
|\scrD(\rho)|\leq\int_1^\infty\frac{|k(\rho w)|dw}{\sqrt{w^2-1}}
\leq \sum_{m\in\scrN} r_\scrL(m)\int_1^\infty \frac{e^{-\sigma\sqrt{m}w}dw}{\sqrt{w^2-1}}=f(\sigma)
\end{equation}
which implies for sufficiently large $\sigma>0$ (in particular it is necessary that $\log\sigma>2\pi|c(\varphi)|$)
\begin{equation}
\frac{|\scrD(\rho)|}{|\log\i\rho-2\pi c(\varphi)|}\leq\frac{f(\sigma)}{\log\sigma-2\pi |c(\varphi)|}<1.
\end{equation}
\end{proof}

Let $h$ be an even function, analytic in a strip
$|\Im\rho|\leq\sigma_0$, for some $\sigma_0>\sigma$, which satisfies
\begin{equation}\label{relaxed h}
 |h(\rho)|\ll(1+|\Re\rho|)^{-5-\delta}
\end{equation}
uniformly in the same strip for some $\delta>0$.
\begin{remark}
We restrict ourselves to a smaller space of test functions here to simplify the presentation of our arguments. It is possible to obtain the trace formula for any test function with uniform decay $|h(\rho)|\ll(1+|\Re\rho|)^{-2-\delta}$. 
\end{remark}
Let $T>0$. Define the box $$B(\sigma,T)=\lbrace\rho\mid|\Re\rho|\leq T,\,|\Im\rho|\leq\sigma\rbrace.$$
\begin{prop}\label{tracethm}
Denote by $n_j=\rho_j^2$, $\rho_j\geq0$, the eigenvalues without counting multiplicities. The new eigenvalues which lie strictly between the $n_j$ are denoted by $\lambda_j^\varphi=(\rho_j^\varphi)^2$. We denote $0>\lambda_0^\varphi=(\rho_0^\varphi)^2$ where $\rho_0^\varphi$ is purely imaginary and $\Im\rho_0^\varphi>0$. Let $\sigma>\Im\rho_0^\varphi$ and $T>0$ s.t. $T\notin\lbrace\rho_j\rbrace_j\cup\lbrace\rho_j^\varphi\rbrace_j$. We have
\begin{equation}\label{tf1}
\begin{split}
&2h(\rho_0^\varphi)-2h(0)+2\sum_{0<\rho_j^\varphi<T}h(\rho_j^\varphi)-2\sum_{0<\rho_j<T}h(\rho_j)\\
=&\frac{1}{2\pi\i}\int_{\partial B(\sigma,T)}h(\rho)\frac{S'_{\varphi}}{S_\varphi}(\rho)d\rho\\
=&\frac{1}{\pi\i}\left\{\int_{-\i\sigma-T}^{-\i\sigma+T}+\int_{-\i\sigma+T}^{\i\sigma+T}\right\}h(\rho)\frac{S'_{\varphi}}{S_\varphi}(\rho)d\rho
\end{split}
\end{equation}
\end{prop}
\begin{proof}
By contour integration and symmetry. This is clear in view of the spectral expansion
\begin{equation}
S_\varphi(\rho)=-\frac{1}{\rho^2}+\sum_{j=1}^{\infty}r(n_j)\left\{\frac{1}{n_j-\rho^2}-\frac{n_j}{n_j^2+1}\right\}-c_0\tan\frac{\varphi}{2}.
\end{equation}
\end{proof}

We may rewrite \eqref{tf1} as
\begin{equation}\label{pretraceB}
\begin{split}
&2h(\rho_0^\varphi)-2h(0)+2\sum_{0<\rho_j^\varphi<T}h(\rho_j^\varphi)-2\sum_{0<\rho_j<T}h(\rho_j)\\
=&\frac{1}{\pi\i}\int_{-\i\sigma-T}^{-\i\sigma+T}\frac{h(\rho)d\rho}{\rho(\log\i\rho -2\pi c(\varphi))}\\
&-\frac{1}{\pi\i}\int_{-\i\sigma-T}^{-\i\sigma+T}h'(\rho)\log\left(1-\frac{\scrD(\rho)}{\log\i\rho-2\pi c(\varphi)}\right)d\rho
+\partial B(T)
\end{split}
\end{equation}
where
\begin{equation}
\begin{split}
\partial B(T)=&\frac{1}{\pi\i}\left[h(r)\log\left(\frac{S_\varphi(r)}{\log(\i r)-2\pi c(\varphi)}\right)\right]_{-\i\sigma-T}^{-\i\sigma+T}\\
&+\frac{1}{\pi\i}\int_{-\i\sigma+T}^{\i\sigma+T}h(\rho)\frac{S'_{\varphi}}{S_\varphi}(\rho)d\rho.
\end{split}
\end{equation}


We have the following fact, analogous to Theorem 12 in \cite{U}.
\begin{prop}\label{vanishing}
There exists an increasing sequence $\lbrace
t_l\rbrace\subset\R_+\setminus(\lbrace\rho_j\rbrace_j\cup\lbrace\rho_j^\varphi\rbrace_j)$
such that $\lim_{l\to\infty}t_l=+\infty$ and
$$\lim_{l\to\infty}\partial B(t_l)=0.$$
\end{prop}

Since the sums and integrals  in \eqref{pretraceB} (where we take
$T=t_l$) converge absolutely as $t_l\to\infty$,
Proposition~\ref{vanishing}, which we will prove in
section~\ref{Sec:Trace2}, gives the general trace formula:
\begin{thm}\label{TF}
Let $h$ be as \eqref{relaxed h}. Let $\sigma>0$ be large enough s.t.
condition \eqref{cond} is satisfied. We have
\begin{equation}\label{pretrace}
\begin{split}
&\sum_{j=0}^\infty \lbrace h(\rho_j^\varphi)-h(\rho_j)\rbrace\\
=&\frac{1}{2\pi\i}\int_{-\i\sigma-\infty}^{-\i\sigma+\infty}\frac{h(\rho)d\rho}{\rho(\log\i\rho-2\pi c(\varphi))}\\
&-\frac{1}{2\pi\i}\int_{-\i\sigma-\infty}^{-\i\sigma+\infty}h'(\rho)\log\left(1-\frac{\scrD(\rho)}{\log\i\rho-2\pi c(\varphi)}\right)d\rho.
\end{split}
\end{equation}
\end{thm}
We call the first term on the RHS of \eqref{pretrace} the "smooth
term", and the second one the "diffractive term".

\section{Proof of Proposition~\ref{vanishing}}\label{Sec:Trace2}

We begin with the following lemma.
\begin{lem}\label{polybound}
There exists an increasing sequence $\lbrace T_n\rbrace\subset\R_+\setminus(\lbrace\rho_j\rbrace_j\cup\lbrace\rho_j^\varphi\rbrace_j)$ such that $\lim_{n\to\infty}T_n=+\infty$ and for $-\sigma\leq w\leq0$ we have
\begin{equation}
|S_\varphi(T_n+\i w)|\ll_\epsilon T_n^{4+\epsilon}.
\end{equation}
\end{lem}
\begin{proof}
We can choose an infinite increasing subsequence of Laplacian eigenvalues $\lbrace n_{k(n)}\rbrace_n$ such that $n_{k(n)+1}-n_{k(n)}=\rho_{k(n)+1}^2-\rho_{k(n)}^2\gg1$. This is because the mean spacing between the norms $\lbrace n_j\rbrace$ is of size $\sqrt{\log n_j}$ if the lattice $\scrL$ is rational and of size $1$ if the lattice is irrational. Recall that between two consecutive eigenvalues $n_{k(n)}=\rho_{k(n)}^2$ and $n_{k(n)+1}=\rho_{k(n)+1}^2$ there is exactly one new eigenvalue $\lambda^\varphi_{k(n)+1}=\chi_{k(n)+1}^2$ and $\chi_{k(n)+1}\in(\rho_{k(n)},\rho_{k(n)+1})\subset\R_{+}$ is a zero of the function $S_{\varphi}(\rho)$, whereas $\rho_{k(n)}$, $\rho_{k(n)+1}$ are singularities of the same function.

So we may choose an infinite sequence
\begin{equation}\label{infseq}
T_{n}=
\begin{cases}
\tfrac{1}{2}(\rho_{k(n)}+\chi_{k(n)+1}),\;\text{if}\;|\chi_{k(n)+1}-\rho_{k(n)}|\geq|\chi_{k(n)}-\rho_{k(n)+1}|\\
\\
\tfrac{1}{2}(\rho_{k(n)+1}+\chi_{k(n)+1}),\;\text{otherwise.}
\end{cases}
\end{equation}
with $|\rho_{k(n)}-\rho_{k(n)+1}|\gg|\rho_{k(n)}+\rho_{k(n)+1}|\asymp T_{n}^{-1}$. Note in particular that for all $\rho_j\in\R_+$,
\begin{equation}\label{lboundTN}
|\rho_j-T_n|\geq\tfrac{1}{4}|\rho_{k(n)}-\rho_{k(n)+1}|\gg T_n^{-1}.
\end{equation}

Let $\mu_{n}(w)=(T_{n}+\i w)^{2}$, $w\in[-\sigma,0]$. We have for any $\epsilon>0$
\begin{equation}
\begin{split}
|S^\varphi(T_n+\i w)|\ll
&\sum_{j=0}^{\infty}r_\scrL(n_j)\left|\frac{1}{n_{j}-\mu_n(w)}-\frac{1}{n_{j}-\i}\right|\\
&\ll_\epsilon |\i-\mu_{N}(w)|\sum_{j=0}^{\infty}\frac{n_{j}^{\epsilon}}{|n_{j}-\mu_{N}(w)||n_{j}-\i|}
\end{split}
\end{equation}
where we have used the bound $r_\scrL(n)\ll_\epsilon n^\epsilon$. Fix $\alpha\in(\epsilon,1)$. We split the sum into a central part satisfying $\inf_{w\in[-\sigma,0]}|n_{j}-\mu_{N}(w)|<n_{j}^{\alpha}$ and a corresponding tail. For convenience we let $I_{n}(n_{j})=\inf_{w\in[-\sigma,0]}|n_{j}-\mu_{n}(w)|$. The first sum is estimated by
\begin{equation}
\begin{split}
&\sum_{I_{n}(n_{j})<n_{j}^{\alpha}}\frac{n_{j}^{\epsilon}}{|n_{j}-\mu_{n}(w)||n_{j}-\i|}\\
\leq&\#\lbrace j\mid I_{n}(n_{j})<n_{j}^{\alpha}\rbrace\, \max_{I_{n}(n_{j})<n_{j}^{\alpha}}\,\sup_{w\in[-\sigma,0]}\left\{\frac{n_{j}^{\epsilon}}{|n_{j}-\mu_{n}(w)||n_{j}-\i|}\right\}.
\end{split}
\end{equation}
Now if $n_{j}>T_{n}^{2}$ then $I_{n}(n_{j})=n_{j}-T_{n}^{2}$. It follows
\begin{equation}
\begin{split}
&\#\lbrace j\mid I_{n}(n_{j})<n_{j}^{\alpha}\rbrace\\
\leq&\#\lbrace j\mid n_{j}\leq T_{n}^{2}\rbrace
+\#\lbrace j\mid n_{j}-n_{j}^{\alpha}<T_{n}^{2}\rbrace.
\end{split}
\end{equation}
Let
\begin{equation}
C(\alpha)=\#\lbrace j\mid n_{j}\leq 2^{1/(1-\alpha)}\rbrace
\end{equation}
and observe that $n_{j}>2^{1/(1-\alpha)}$ implies $n_{j}^{\alpha-1}<\tfrac{1}{2}$. So $n_{j}>2^{1/(1-\alpha)}$ together with $n_{j}(1-n_{j}^{\alpha-1})<T_{n}^{2}$ implies
\begin{equation}
n_{j}<2n_{j}(1-n_{j}^{\alpha-1})<2T_n^2.
\end{equation}
Hence
\begin{equation}
\begin{split}
&\#\lbrace j\mid n_{j}(1-n_{j}^{\alpha-1})<T_{n}^{2}\rbrace\\
\leq&\,\#\lbrace j\mid n_{j}\leq2^{1/(1-\alpha)},\;n_{j}(1-n_{j}^{\alpha-1})<T_{n}^{2}\rbrace\\
&+\#\lbrace j\mid n_{j}>2^{1/(1-\alpha)},\;n_{j}(1-n_{j}^{\alpha-1})<T_{n}^{2}\rbrace\\
\leq&\,C(\alpha)+\#\lbrace j\mid2^{1/(1-\alpha)}<n_{j}<2T_{n}^{2}\rbrace\\
\ll&\; T_n^2.
\end{split}
\end{equation}
It follows that
\begin{equation}
\#\lbrace j\mid I_{n}(n_{j})<n_{j}^{\alpha}\rbrace\ll T_{n}^{2}.
\end{equation}
By the same observations as above we see that $I(n_{j})<n_{j}^{\alpha}$ implies $n_{j}\leq\max\lbrace2^{1/(1-\alpha)},2T_{n}^{2}\rbrace$. Also for any $j\geq0$ we have (cf. \eqref{lboundTN}) $$|\rho_{j}-T_{n}|\geq\tfrac{1}{4}|\rho_{k(n)}-\rho_{k(n)+1}|\gg T_{n}^{-1}$$ which implies
\begin{equation}
\begin{split}
|n_{j}-\mu_{n}(w)|=|\rho_{j}^{2}-(T_{n}+\i w)^{2}|
=&\;|\rho_{j}-T_{n}-\i w||\rho_{j}+T_{n}+\i w|\\
\geq&\;|\rho_{j}-T_{n}|(\rho_{j}+T_{n})\\
\gg&\;1.
\end{split}
\end{equation}
Since $|n_{j}-\i|\geq1$ we have
\begin{equation}
\max_{I_{n}(n_{j})<n_{j}^{\alpha}}\,\sup_{w\in[-\sigma,0]}\left\{\frac{n_{j}^{\epsilon}}{|n_{j}-\mu_{n}(w)||n_{j}-\i|}\right\}
\ll T_{n}^\epsilon.
\end{equation}

The tail can be bounded as follows
\begin{equation}
\begin{split}
\sum_{I_n(n_{j})\geq n_{j}^{\alpha}}\frac{n_{j}^{\epsilon}}{|n_{j}-\mu_{n}(w)||n_{j}-\i|}
\leq&\sum_{I_n(n_{j})\geq n_{j}^{\alpha}}\frac{n_{j}^{\epsilon-\alpha}}{|n_{j}-\i|}\\
\leq&\sum_{j=0}^{\infty}\frac{n_{j}^{\epsilon-\alpha}}{|n_{j}-\i|}=O(1).
\end{split}
\end{equation}
Finally note that $|\mu_{n}(w)-\i|\ll T_n^2$.
\end{proof}

Recall
\begin{equation}
\begin{split}
\partial B(T_n)=&\frac{1}{\pi\i}\left[h(r)\log\left(1-\frac{\scrD(r)}{\log(\i r)-2\pi c(\varphi)}\right)\right]_{-\i\sigma-T_n}^{-\i\sigma+T_n}\\
&+\frac{1}{\pi\i}\int_{-\i\sigma+T_n}^{\i\sigma+T_n}h(\rho)\frac{S'_{\varphi}}{S_\varphi}(\rho)d\rho.
\end{split}
\end{equation}
We know that $\lim_{n\to\infty} B(T_n)$ exists (for any test function $h$ which satisfies the uniform bound $|h(\rho)|\ll(1+|\Re\rho|)^{-2-\delta}$ in the strip $|\Im\rho|\leq\sigma$ -- the decay which is required by Weyl's law to ensure that the trace converges absolutely) and we want to prove that the limit is zero for any test function which satisfies the uniform bound
\begin{equation}
|h(\rho)|\ll(1+|\Re\rho|)^{-5-\delta}
\end{equation}
in the strip $|\Im\rho|\leq\sigma$.

For the first term we have, in view of $|\scrD(\rho)|\leq f(\sigma)$ along $\Im\rho=-\sigma$,
\begin{equation}
\left|\log\left(1-\frac{\scrD(-\i\sigma\pm T_n)}{\log(\sigma\pm\i T_n)-2\pi c(\varphi)}\right)\right|\ll \frac{1}{\log T_n}
\end{equation}
which implies that this term vanishes as $n\to\infty$.

For the integral an integration by parts gives
\begin{equation}
\begin{split}
\int_{-\i\sigma+T_n}^{\i\sigma+T_n}h(\rho)\frac{S'_{\varphi}}{S_\varphi}(\rho)d\rho
=&\left[h(\rho)\log S_\varphi(\rho)\right]_{-\i\sigma+T_n}^{\i\sigma+T_n}\\
&-\int_{-\i\sigma+T_n}^{\i\sigma+T_n}h'(\rho)\log S_\varphi(\rho)d\rho.
\end{split}
\end{equation}
To see that the first term vanishes as $n\to\infty$, observe that the identity \eqref{specid} and the bound \eqref{o1bound} imply
\begin{equation}
\begin{split}
|\log S(\pm\i\sigma+T_n)|=\;&|\log S(-\i\sigma \mp T_n)| \\
\leq\; &|\log|S(\pm\i\sigma+T_{n})||+|\arg S(\pm\i\sigma+T_n)| \\
=\; &\log\log T_n+O(1),
\end{split}
\end{equation}
where we used $|\arg S(\pm\i\sigma+T_n)|\ll1$ as $n\to\infty$. Similarly we see $$\int_{-\i\sigma+T_n}^{\i\sigma+T_n}h'(\rho)\log S_\varphi(\rho)d\rho=\int_{-\i\sigma+T_n}^{\i\sigma+T_n}h'(\rho)\log |S_\varphi(\rho)|d\rho+O(T_n^{-5}).$$
We have the calculation
\begin{equation}
\begin{split}
&\int_{T_n}^{\i\sigma+T_n}h'(\rho)\log |S_\varphi(\rho)|d\rho\\
\stackrel{\rho\to-\rho}{=}&-\int_{-T_n}^{-\i\sigma-T_n}h'(-\rho)\log |S_\varphi(-\rho)|d\rho\\
\stackrel{\rho\to-\bar{\rho}}{=}&-\int_{T_n}^{-\i\sigma+T_n}h'(\bar{\rho})\log |S_\varphi(\bar{\rho})|d\rho\\
=&\int_{-\i\sigma+T_n}^{T_n}h'(\bar{\rho})\log |S_\varphi(\rho)|d\rho
\end{split}
\end{equation}
where we used $S_\varphi(\bar{\rho})=\overline{S_\varphi(\rho)}$, and so the term
\begin{equation}\label{btermidentity}
\int_{-\i\sigma+T_n}^{\i\sigma+T_n}h'(\rho)\log |S_\varphi(\rho)|d\rho=\int_{-\i\sigma+T_n}^{T_n}\lbrace h'(\rho)+h'(\bar{\rho})\rbrace\log |S_\varphi(\rho)|d\rho
\end{equation}
converges to a limit as $n\to\infty$.

To obtain the result we require two lemmas. The first lemma constructs an even test function which is analytic in a strip and the real part of whose derivative satisfies a certain polynomial lower bound in $T_n$ on the line segment $[-\i\sigma+T_n,T_n]$.
\begin{lem}\label{testfunction}
Choose $\sigma_0>\sigma$. Let $$h_5(\rho)=\frac{-1}{(\rho^2+\sigma_0^2)^2}.$$ We have for $t\in[-\sigma_0,0]$ and for all sufficiently large $n$
\begin{equation}
\Re h'_5(T_n+\i t)=|\Re h'_5(T_n+\i t)|\gg \frac{1}{T_n^5}
\end{equation}
as $n\to\infty$.
\end{lem}
\begin{proof}
We have $$h_5'(\rho)=\frac{4\rho}{(\rho^2+\sigma_0^2)^3}.$$
Let $t\in[-\sigma_0,0]$. A simple calculation gives
\begin{equation}
\begin{split}
|\Re h_5'(T_n+\i t)|=&\left|\Re\left\{\frac{4(T_n+\i t)(T_n^2-t^2+\sigma_0^2-2\i T_n t)^3}{((T_n^2-t^2+\sigma_0^2)^2+4T_n^2t^2)^3}\right\}\right|\\
=&\Re\left\{\frac{4(T_n+\i t)(T_n^2-t^2+\sigma_0^2-2\i T_n t)^3}{((T_n^2-t^2+\sigma_0^2)^2+4T_n^2t^2)^3}\right\}\gg \frac{1}{T_n^5}
\end{split}
\end{equation}
as $n\to\infty$.
\end{proof}

The second lemma gives a bound on $\log|S_\varphi|$ averaged along the line segment $[-\i\sigma+T_n,T_n]$.
\begin{lem}\label{avbound}
We have the following bound
\begin{equation}
\left|\int_{-\i\sigma+T_{n}}^{T_{n}}\log |S_\varphi(\rho)|d\rho\right|\ll T_{n}^5.
\end{equation}

\end{lem}
\begin{proof}
We know there exists a constant $c>0$ such that for all $n$ and $w\in[-\sigma,0]$ we have $$|S_\varphi(T_n+\i w)|<cT_n^5.$$ In Lemma \ref{testfunction} we prove the existence of a test function $h_5$ which is analytic in the strip $|\Im\rho|\leq\sigma$, satisfies the uniform bound $|h_5(\rho)|\ll(1+|\Re\rho|)^{-4}$ in this strip and in addition $h_5(\bar{\rho})=\overline{h_5(\rho)}$ and $\Re h'_5 (T_{n}+\i w)=|\Re h'_5 (T_{n}+\i w)|\gg T_{n(l)}^{-5}$. We thus have
\begin{equation}
\begin{split}
&T_{n}^{-5}\left|\int_{-\i\sigma+T_{n}}^{T_{n}}\log |S_\varphi(\rho)|d\rho\right|\\
\leq \;&T_{n}^{-5}\int_{-\sigma}^0 -\log(c^{-1}T_{n}^{-5}|S_\varphi(T_{n}+\i w)|)dw+O(T_{n}^{-5}\log T_{n})\\
\ll \;& -\int_{-\sigma}^0\Re h'_5 (T_{n}+\i w)\log(c^{-1}T_{n}^{-5}|S_\varphi(T_{n}+\i w)|)dw\\
&+O(T_{n}^{-5}\log T_{n})\\
\ll_\epsilon \;&1
\end{split}
\end{equation}
because $|h_5(\rho)|\ll(1+|\Re\rho|)^{-4}$ uniformly in $|\Im\rho|\leq\sigma$ and therefore $$\lim_{n\to\infty}\int_{-\i\sigma+T_{n}}^{T_{n}}\Re h'_5(\rho)\log |S_\varphi(\rho)|d\rho$$ exists.
\end{proof}
We obtain $$\left|\int_{-\i\sigma+T_{n}}^{\i\sigma+T_{n}}h'(\rho)\log |S_\varphi(\rho)|d\rho\right|\ll T_{n}^{-\delta}$$ in view of the identity \eqref{btermidentity}.
We also used that by Cauchy's theorem the analyticity and decay of $h$ in $|\Im\rho|\leq\sigma_0$, where $\sigma_0>\sigma$, imply the analyticity of $h'$ in $|\Im\rho|\leq\sigma$ and the uniform decay $$|h'(\rho)|\ll(1+|\Re\rho|)^{-5-\delta}$$ in the same strip. It follows that $$\lim_{n\to\infty} \partial B(T_{n})=0$$ which proves Proposition~\ref{vanishing}.

\section{Proof of Proposition~\ref{main bound}}\label{sec:using trace}

\subsection{}
We want to apply the trace formula in order to obtain information about the average spacing between new eigenvalues and old eigenvalues. Let $h(\rho)=e^{-\beta\rho^2}$, for small $\beta>0$. Upon dividing through by $\beta$ we can rewrite the l.h.s. of the trace formula \eqref{pretrace} as
\begin{equation}\label{approx}
\begin{split}
\frac{1}{\beta}\sum_{j=0}^\infty\lbrace e^{-\beta\lambda_j^\varphi}-e^{-\beta n_j}\rbrace
=&\sum_{j=0}^\infty e^{-\beta\lambda_j^\varphi}\frac{1-e^{-\beta(n_j-\lambda_j^\varphi)}}{\beta}\\
=&\sum_{j=0}^\infty d_j e^{-\beta\lambda_j^\varphi}+O(\beta^{-1/2})
\end{split}
\end{equation}
where $d_j=n_j-\lambda_j^\varphi>0$. The last line follows from the following lemma.
\begin{lem}\label{approxlemma}
We have the bound
\begin{equation}
\sum_{j=0}^\infty d_j e^{-\beta\lambda_j^\varphi}\left(1-\frac{1-e^{-\beta d_j}}{\beta d_j}\right)
\ll\beta^{-1/2}
\end{equation}
\end{lem}
\begin{proof}

For $x>0$ we have the inequality $$0<1-\frac{1-e^{-x}}{x}<x$$ and the bound $d_j\ll n_j^{1/4}$ for $j\geq1$ (cf. the greedy algorithm in \cite{RU}, p. 7). It follows
\begin{equation}
\begin{split}
\sum_{j=0}^\infty d_j e^{-\beta\lambda_j^\varphi}\left(1-\frac{1-e^{-\beta d_j}}{\beta d_j}\right)
<&\beta\sum_{j=0}^\infty d_j^2 e^{-\beta \lambda_j^\varphi}\\
\ll& \beta\sum_{j=1}^\infty n_j^{1/2}e^{-\beta\lambda_j^\varphi}+\beta e^{-\beta\lambda_0^\varphi}\\
<&\beta\sum_{j=0}^\infty n_{j+1}^{1/2}e^{-\beta n_j}+\beta e^{-\beta\lambda_0^\varphi}
\end{split}
\end{equation}
and the bound $N_\varphi(x)\ll x$ permits us to bound the sum by the following integral:
\begin{equation}
\beta\sum_{j=0}^\infty n_{j+1}^{1/2}e^{-\beta n_j}\ll \beta\int_{0}^\infty x^{1/2}e^{-\beta x}dx
\ll\beta^{-1/2}.
\end{equation}
\end{proof}

\subsection{The smooth term}
We have the following bound on the smooth term.
\begin{prop}\label{smoothbound}
As $\beta\searrow0$
\begin{equation}
\frac{1}{2\pi}\left|\int_{-\i\sigma-\infty}^{-\i\sigma+\infty}\frac{e^{-\beta\rho^2}d\rho}{\rho(\log(\i\rho)-2\pi
c(\varphi))}\right| \ll \frac{1}{\log\frac 1{\beta}}.
\end{equation}
\end{prop}
\begin{proof}
Denote by $C_\delta$ the contour following a semicircle in the lower halfplane centered at the origin of radius $\delta$, where $e^{2\pi c(\varphi)}>\delta>0$, starting from $-\delta$ and finishing at $\delta$. By shifting the contour across the pole at $\rho=-\i e^{2\pi c(\varphi)}$ to the real line we obtain
\begin{equation}\label{id1}
\begin{split}
&\frac{1}{2\pi\i}\int_{-\i\sigma-\infty}^{-\i\sigma+\infty}\frac{e^{-\beta\rho^2}d\rho}{\rho\log(\i\rho e^{-2\pi c(\varphi)})}\\
=&e^{\beta e^{4\pi c(\varphi)}}+\frac{1}{2\pi\i}\left\{\int_{C_\delta}+\int_{\R\setminus(-\delta,\delta)}\right\}\frac{e^{-\beta\rho^2}d\rho}{\rho\log(\i\rho e^{-2\pi c(\varphi)})}.
\end{split}
\end{equation}
Note that the integral over the semicircle vanishes as $\delta\to0$.

We may pick the branch of the complex logarithm in such a way that $\arg(x)=\pi/2$ if $x<0$ and $\arg(x)=3\pi/2$ if $x>0$.  Then for real $\rho\neq0$
\begin{equation}
\begin{split}
&\frac{1}{\log(\i\rho e^{-2\pi c(\varphi)})}=\frac{1}{\log(|\rho|e^{-2\pi c(\varphi)})+\i(\tfrac{\pi}{2}+\arg(\rho e^{-2\pi c(\varphi)}))}\\
=&\frac{\log(|\rho|e^{-2\pi c(\varphi)})-\i(\tfrac{\pi}{2}+\arg(\rho e^{-2\pi c(\varphi)}))}{\log^2(|\rho| e^{-2\pi c(\varphi)})+\pi^2/4}
\end{split}
\end{equation}
and it follows that
\begin{equation}
\begin{split}
&\frac{1}{2\pi\i}\int_{\R\setminus(-\delta,\delta)}\frac{e^{-\beta\rho^2}d\rho}{\rho\log(\i\rho e^{-2\pi c(\varphi)})}\\
=&-\frac{1}{2\pi}\int_{\R\setminus(-\delta,\delta)}e^{-\beta\rho^2}\frac{\tfrac{\pi}{2}+\arg(\rho e^{-2\pi c(\varphi)})}{\rho(\log^2(|\rho|e^{-2\pi c(\varphi)})+\frac{\pi^2}{4})}d\rho\\
=&-\frac{1}{2}\int_\delta^\infty \frac{e^{-\beta\rho^2} d\rho}{\rho(\log^2(\rho e^{-2\pi c(\varphi)})+\tfrac{\pi^2}{4})}\\
=&-\frac{1}{2}\int_{e^{-2\pi c(\varphi)}\delta}^\infty \frac{e^{-\beta e^{4\pi c(\varphi)} r^2} dr}{r(\log^2 r+\tfrac{\pi^2}{4})}\\
=&-\frac{1}{2}\int_{-2\pi c(\varphi)+\log\delta}^\infty \frac{e^{-\beta e^{4\pi c(\varphi)} e^{2t}} dt}{t^2+\frac{\pi^2}{4}}\\
\to&-\frac{1}{2}\int_{-\infty}^\infty \frac{e^{-\beta e^{4\pi c(\varphi)}e^{2t}} dt}{t^2+\frac{\pi^2}{4}}\qquad\text{as $\delta\to0$.}
\end{split}
\end{equation}
Since $$\int_{-\infty}^\infty \frac{dt}{t^2+\frac{\pi^2}{4}}=2,$$ we obtain in view of \eqref{id1}
\begin{equation}
\begin{split}
&\frac{1}{2\pi\i}\int_{-\i\sigma-\infty}^{-\i\sigma+\infty}\frac{e^{-\beta\rho^2}d\rho}{\rho\log(\i\rho e^{-2\pi c(\varphi)})}\\
=&e^{\beta e^{4\pi c(\varphi)}}-1+\frac{1}{2}\int_{-\infty}^\infty\frac{1-e^{-\beta e^{4\pi c(\varphi)} e^{2t}}dt}{t^2+\tfrac{\pi^2}{4}}.
\end{split}
\end{equation}
Let $\gamma=e^{4\pi c(\varphi)}\beta$. We proceed by dividing the integral on the r.h.s. into two integrals over the ranges $(-\infty,\tfrac{1-\epsilon}{2}|\log\gamma|)$ and $[\tfrac{1-\epsilon}{2}|\log\gamma|,\infty)$ for some small $\epsilon>0$.

We then bound the first integral as follows
\begin{equation}
\int_{-\infty}^{\tfrac{1-\epsilon}{2}|\log\gamma|}\frac{1-e^{-\gamma e^{2t}}dt}{t^2+\tfrac{\pi^2}{4}}
\ll \gamma\int_{-\infty}^{\tfrac{1-\epsilon}{2}|\log\gamma|}e^{2t}dt=\tfrac{1}{2}\gamma^{\epsilon}
\end{equation}
where we note that $$|1-e^{-\gamma e^{2t}}|\ll \gamma e^{2t},$$ because $t<\tfrac{1-\epsilon}{2}|\log\gamma|$ implies $\gamma e^{2t}<\gamma^\epsilon$.

We bound the second integral by
\begin{equation}
\frac{1}{2}\int_{\tfrac{1-\epsilon}{2}|\log\gamma|}^{\infty}\frac{|1-e^{-\gamma e^{2t}}|dt}{t^2+\tfrac{\pi^2}{4}}
<\int_{\tfrac{1-\epsilon}{2}|\log\gamma|}^{\infty}t^{-2}dt=\frac{2}{(1-\epsilon)|\log\gamma|}.
\end{equation}
\end{proof}

\subsection{The diffractive term}
We continue with the bound on the diffractive term.
\begin{prop}\label{diffbound}
Let $h(\rho)=e^{-\beta\rho^2}$. As $\beta\searrow0$
\begin{equation}
\left|\int_{-\i\sigma-\infty}^{-\i\sigma+\infty}h'(\rho)\log\left(1-\frac{\scrD(\rho)}{\log\i\rho}\right)d\rho\right|
\ll \frac{1}{\log\tfrac{1}{\beta}}.
\end{equation}
\end{prop}
\begin{proof}
\eqref{cond} allows us to estimate
\begin{equation}
\begin{split}
&\left|\int_{-\infty}^{\infty}h'(-\i\sigma+s)\log\left(1-\frac{\scrD(-\i\sigma+s)}{\log\i(-\i\sigma+s)-2\pi c(\varphi)}\right)ds\right|\\
\ll &\int_{-\infty}^{\infty}|h'(-\i\sigma+s)|\frac{|\scrD(-\i\sigma+s)|}{|\log\i(-\i\sigma+s)-2\pi c(\varphi)|}ds.
\end{split}
\end{equation}
We have
\begin{equation}
|h'(-\i\sigma+s)|=2\beta|-\i\sigma+s| |e^{-\beta(-\i\sigma+s)^2}|
\leq 2\beta(\sigma+|s|)e^{\beta\sigma^2-\beta s^2 }
\end{equation}
and
\begin{equation}
|\scrD(-\i\sigma+s)|\leq\sum_{m\in\scrN} r_\scrL(m)\int_1^\infty \frac{e^{-\sigma\sqrt{m}w}dw}{\sqrt{w^2-1}}=f(\sigma)
\end{equation}
and finally
\begin{equation}
|\log(\sigma+\i s)|\geq\tfrac{1}{2}\log(\sigma^2+s^2).
\end{equation}
We continue our estimate as follows (recall $\sigma>\max\lbrace1,e^{2\pi c(\varphi)}\rbrace$)
\begin{equation}
\begin{split}
&\int_{-\infty}^{\infty}|h'(-\i\sigma+s)|\frac{|\scrD(-\i\sigma+s)|}{|\log\i(-\i\sigma+s)-2\pi c(\varphi)|}ds\\
\leq &\;8\beta e^{\beta\sigma^2}f(\sigma)\int_{0}^{\infty}\frac{(\sigma+s)e^{-\beta s^2}ds}{\log(\sigma^2+s^2)-4\pi c(\varphi)}
\end{split}
\end{equation}
and the integral is bounded by
\begin{equation}
\begin{split}
\int_{0}^{\infty}\frac{(\sigma+s)e^{-\beta s^2}ds}{\log(\sigma^2+s^2)-4\pi c(\varphi)}
<&\frac{\sigma}{(2\log\sigma-4\pi c(\varphi))\beta^{1/2}}\int_0^\infty e^{-w^2}dw\\
&+\int_{0}^{\infty}\frac{se^{-\beta s^2}ds}{\log(\sigma^2+s^2)-4\pi c(\varphi)}.
\end{split}
\end{equation}
Let $\gamma=e^{4\pi c(\varphi)}\beta$ and $\xi=e^{-2\pi c(\varphi)}\sigma$. We bound the second integral on the r.h.s. as follows
\begin{equation}
\begin{split}
\int_{0}^{\infty}\frac{se^{-\beta s^2}ds}{\log(\sigma^2+s^2)-4\pi c(\varphi)}
=&\frac{1}{\beta}\int_{0}^{\infty}\frac{we^{-w^2}ds}{\log(\sigma^2+\tfrac{w^2}{\beta})-4\pi c(\varphi)}\\
=&\frac{1}{\beta}\int_{0}^{\infty}\frac{we^{-w^2}ds}{\log(\xi^2+\tfrac{w^2}{\gamma})}\\
\ll&\frac{1}{\beta\log\xi^2}\int^{\gamma^\epsilon}_0 we^{-w^2}dw\\
&+\frac{1}{\beta(1-2\epsilon)\log\frac{1}{\gamma}}\int_{\gamma^\epsilon}^\infty we^{-w^2}dw
\end{split}
\end{equation}
for small $\epsilon>0$. The last line follows, since for $w\geq\gamma^\epsilon$ we have (assuming $\gamma<1$)
\begin{equation}
\begin{split}
1+\frac{\log(\gamma\sigma^2+w^2)}{\log\frac{1}{\gamma}}&\geq1+\frac{2\log w+\log(1+\gamma\sigma^2/w^2)}{\log(\frac{1}{\gamma})}\\
&\geq1-2\epsilon+O\left(\frac{\gamma^{1-2\epsilon}}{\log\frac{1}{\gamma}}\right).
\end{split}
\end{equation}
Since $$\int^{\gamma^\epsilon}_0 we^{-w^2}dw=O(\gamma^{2\epsilon})$$ we have $$\int_{-\infty}^{\infty}|h'(-\i\sigma+s)|\frac{|\scrD(-\i\sigma+s)|}{|\log\i(-\i\sigma+s)|}ds\ll\frac{1}{\log(\frac{1}{\gamma})}+O(\gamma^{2\epsilon}).$$
\end{proof}

\section{The three-dimensional case}\label{sec:3d}

The three-dimensional case is very different. Consider the three-dimensional flat torus $\T^3=\R^3/2\pi\scrL^3_0$, where $\scrL^3_0=\Z(1/ab,0,0)\oplus\Z(0,a,0)\oplus\Z(0,0,b)$ for some $a,b>0$. Denote by $\scrL^3$ the dual lattice of $\scrL^3_0$. The eigenvalues of the Laplacian on $\T^3$ are the norms of the vectors of the dual lattice $\scrL^3$. We denote the set of norms of the dual lattice vectors by $\scrN^3$ and the multiplicity of an eigenvalue $n\in\scrN^3$ is denoted by
\begin{equation}
r_3(n)=\#\lbrace\xi\in\scrL^3:|\xi|^2=n\rbrace.
\end{equation}

Let $\eta_j^\varphi$ be the perturbed eigenvalues of the point
scatterer on $\T^3$ and $\eta_j$ the unperturbed eigenvalues counted without
multiplicity (the norms).   The ordering is
\begin{equation}
\eta_0^\varphi<0=\eta_0<\eta_1^\varphi <\eta_1<\dots
<\eta_j^\varphi<\eta_j.
\end{equation}

Our main result is the following. Let $d_j:=\eta_j-\eta_j^\varphi>0$.
\begin{prop}\label{main3d}
We have as $\beta\searrow0$
\begin{equation}
\sum_{j=0}^{\infty}d_j e^{-\beta\eta_j^\varphi} = \frac{1}{2\beta}+O(\beta^{-3/4}).
\end{equation}
\end{prop}

\subsection{The Green's function in three dimensions}
The free Green's function on $\R^3$ is given by (cf. \cite{Z}, p. 842, eq. (3.4))
\begin{equation}
g_\eta(x,x_0)=\frac{e^{-\i\rho|x-x_0|}}{4\pi|x-x_0|}, \qquad \Im\rho<0, \qquad \rho^2=\eta.
\end{equation}

We periodise to obtain the Green's function on $\T^3$.
\begin{equation}
G_\eta(x,x_0)=\frac{1}{4\pi}\sum_{n\in\scrL^3}\frac{e^{-\i\rho|x-x_0+n|}}{|x-x_0+n|}
\end{equation}

In particular the deficiency elements are given by
\begin{equation}
G_{\pm\i}(x,x_0)=\frac{1}{4\pi}\sum_{n\in\scrL^3}\frac{e^{-\frac{|x-x_0+n|}{\sqrt{2}}}\cos\left(\frac{|x-x_0+n|}{\sqrt{2}}\right)}{|x-x_0+n|}
\end{equation}
where we note that $\pm\i=\left(\frac{\pm1-\i}{\sqrt{2}}\right)^2$.

The spectral function is given by
\begin{equation}\label{spectral}
S_3^\varphi(\rho)=\lim_{x\to x_0}(G_\lambda-\Re{G_\i})(x,x_0)-\tan\frac{\varphi}{2}=\frac{-\i\rho}{4\pi}+D^\varphi_3(\rho)
\end{equation}
where
\begin{equation}\label{diffractive}
D^\varphi_3(\rho)=-\tan\frac{\varphi}{2}+\frac{1}{4\pi\sqrt{2}}+\frac{1}{4\pi}\sum_{n\in\scrN^3} r_3(n)\frac{e^{-\i\rho n}-e^{-\frac{n}{\sqrt{2}}}\cos\left(\frac{n}{\sqrt{2}}\right)}{n}
\end{equation}
where the first and second terms in \eqref{spectral} and \eqref{diffractive} come from the regularisation (where $h=x-x_0$)
\begin{equation}
\begin{split}
&\lim_{h\to0}\frac{e^{-\i\rho|h|}
-e^{-\frac{|h+n|}{\sqrt{2}}}\cos\left(\frac{|h+n|}{\sqrt{2}}\right)}{4\pi|h|}\\
&=\lim_{h\to0}\frac{1-\i\rho|h|+O(|h|^2)-(1-\frac{|h|}{\sqrt{2}})(1+O(|h|^2))}{4\pi|h|}\\
&=\frac{-\i\rho}{4\pi}+\frac{1}{4\pi\sqrt{2}}.
\end{split}
\end{equation}

\subsection{The trace formula in three dimensions}

We require the following lemma.
\begin{lem}
For sufficiently large $\sigma>|\Im\rho_0^\varphi|$ we have for $\Im\rho=-\sigma$
\begin{equation}\label{crit3d}
\frac{4\pi|D^\varphi_3(\rho)|}{|\rho|}<1.
\end{equation}
\end{lem}
\begin{proof}
For $\Im\rho=-\sigma$ it can easily be seen from \eqref{diffractive} that $|D_3^\varphi(\rho)|=O(1)$. Furthermore $|\rho|\geq\sigma$. So \eqref{crit3d} certainly holds for sufficiently large $\sigma$.
\end{proof}

We have the following trace formula for a point scatterer in three dimensions.
\begin{thm}\label{TF3d}
Let $h$ be as above and $\sigma>|\Im\rho_0^\varphi|$ large enough such that \eqref{crit3d} is satisfied. Then we have
\begin{equation}
\begin{split}
&\sum_{j=0}^\infty\lbrace h(\rho_j^\varphi)-h(\rho_j)\rbrace\\
=&\tfrac{1}{2}h(0)+\frac{1}{2\pi\i}\int_{-\i\sigma-\infty}^{-\i\sigma+\infty}h'(\rho)\log\left(1+\frac{4\pi\i D^\varphi_3(\rho)}{\rho}\right)d\rho
\end{split}
\end{equation}
\end{thm}
\begin{proof}
Following the argument in the proof of the trace formula for two dimensions we obtain for $\sigma>|\Im\rho_0^\varphi|$ the trace formula (an analogue of Krein's famous trace formula)
\begin{equation}
\sum_{j=0}^\infty\lbrace h(\rho_j^\varphi)-h(\rho_j)\rbrace
=\frac{1}{2\pi\i}\int_{-\i\sigma-\infty}^{-\i\sigma+\infty}h'(\rho)\log S_3(\rho)d\rho.
\end{equation}
In view of \eqref{spectral} we rewrite the r. h. s. as
\begin{equation}
\begin{split}
&\frac{1}{2\pi\i}\int_{-\i\sigma-\infty}^{-\i\sigma+\infty}h'(\rho)\log S_3(\rho)d\rho\\
=&\frac{1}{2\pi\i}\int_{-\i\sigma-\infty}^{-\i\sigma+\infty}h'(\rho)\log\left(\frac{-\i\rho}{4\pi}\right)d\rho\\
&+\frac{1}{2\pi\i}\int_{-\i\sigma-\infty}^{-\i\sigma+\infty}h'(\rho)\log \left(1+\frac{4\pi\i D_3^\varphi(\rho)}{\rho}\right)d\rho
\end{split}
\end{equation}
and the first term can be evaluated by integration by parts and shifting the contour
\begin{equation}
\begin{split}
&\frac{1}{2\pi\i}\int_{-\i\sigma-\infty}^{-\i\sigma+\infty}h'(\rho)\log\left(\frac{-\i\rho}{4\pi}\right)d\rho\\
=&\frac{1}{2\pi\i}\int_{-\i\sigma-\infty}^{-\i\sigma+\infty}\frac{h(\rho)d\rho}{\rho}\\
=&\frac{1}{2\pi\i}\int_{C_\delta}\frac{h(\rho)d\rho}{\rho}
+\underbrace{\frac{1}{2\pi\i}\int_{\R\setminus(-\delta,\delta)}\frac{h(\rho)d\rho}{\rho}}_{=0}\\
=&\tfrac{1}{2}h(0)
\end{split}
\end{equation}
where we recall that for some small $\delta$ the contour $C_\delta$ denotes the lower semicircle connecting $-\delta$ and $\delta$ on the real line.
\end{proof}

\subsection{Proof of Proposition~\ref{main3d}}

Let $h(\rho)=e^{-\beta\rho^2}$.
In this case the trace formula gives us
\begin{equation}
\begin{split}
&\sum_{j=0}^\infty \lbrace e^{-\beta\eta_j^\varphi}-e^{-\beta\eta_j}\rbrace\\
=&\frac{1}{2}-\frac{\beta}{\pi\i}\int_{-\i\sigma-\infty}^{-\i\sigma+\infty}\rho e^{-\beta\rho^2}\log\left(1+\frac{4\pi\i D^\varphi_3(\rho)}{\rho}\right)d\rho
\end{split}
\end{equation}
and in view of \eqref{crit3d} we have for $\Im\rho=-\sigma$ the bound
\begin{equation}
\left|\log\left(1+\frac{4\pi\i D^\varphi_3(\rho)}{\rho}\right)\right|\ll \frac{|D^\varphi_3(\rho)|}{|\rho|}
\end{equation}
which implies
\begin{equation}
\begin{split}
&\int_{-\i\sigma-\infty}^{-\i\sigma+\infty}|\rho||e^{-\beta\rho^2}|\left|\log\left(1+\frac{4\pi\i D^\varphi_3(\rho)}{\rho}\right)\right||d\rho|\\
\ll &\int_{-\i\sigma-\infty}^{-\i\sigma+\infty}|e^{-\beta\rho^2}||D^\varphi_3(\rho)||d\rho|\\
\ll &\; e^{\beta\sigma^2}\int_{-\infty}^\infty e^{-\beta t^2}dt=O\left(\frac{1}{\sqrt{\beta}}\right)
\end{split}
\end{equation}
and therefore
\begin{equation}
\frac{1}{\beta}\sum_{j=0}^\infty \lbrace e^{-\beta\eta_j^\varphi}-e^{-\beta\eta_j}\rbrace=\frac{1}{2\beta}
+O\left(\frac{1}{\sqrt{\beta}}\right).
\end{equation}

An analogue of the greedy algorithm is required for the proof of Proposition~\ref{main3d}. We state this as a lemma.
\begin{lem}\label{greedy3d}
Let $d_j:=\eta_j-\eta_j^\varphi$. We have the bound
\begin{equation}
d_j\ll\eta_j^{1/8}.
\end{equation}
\end{lem}
\begin{proof}
Recall that each $\eta_j$ is of the form $q(m,n,k)=am^2+bn^2+ck^2$ for real numbers $a,b,c>0$ and integers $m,n,k$. We need to show that for each $j$ we can pick $m,n,k$ such that $$|\eta_j^\varphi-q(m,n,k)|\ll\eta_j^{1/8}$$ where the implied constant does not depend on our choice. Let $s_1=\eta_j^\varphi-am^2$. Let $m=\left\lfloor \sqrt{\eta_j^\varphi/a}\right\rfloor$ and observe that $s_1=\eta_j^\varphi-am^2\ll(\eta_j^\varphi)^{1/2}<\eta_j^{1/2}$. Now let $s_2=s_1-bn^2$ and choose $n=\left\lfloor \sqrt{s_1/b}\right\rfloor$ so that $s_2=s_1-bn^2\ll s_1^{1/2}$. Finally choose $k=\left\lfloor \sqrt{s_2/c}\right\rfloor$ such that $s_2-ck^2\ll s_2^{1/2}$. With the above choices of $m,n,k$ we have $$|\eta_j^\varphi-q(m,n,k)|=s_2-ck^2\ll s_2^{1/2}\ll s_1^{1/4}\ll \eta_j^{1/8}.$$
\end{proof}
The following lemma implies Proposition~\ref{main3d}.
\begin{lem}\label{approx3d}
We have the following identity
\begin{equation}
\sum_{j=0}^\infty d_j e^{-\beta\eta_j^\varphi}=\frac{1}{\beta}\sum_{j=0}^\infty \lbrace e^{-\beta\eta_j^\varphi}-e^{-\beta\eta_j}\rbrace+O(\beta^{-3/4})
\end{equation}
\end{lem}
\begin{proof}
It is sufficient to prove the bound
\begin{equation}
\sum_{j=0}^\infty d_j e^{-\beta\lambda_j^\varphi}\left(1-\frac{1-e^{-\beta d_j}}{\beta d_j}\right)
\ll\beta^{-3/4}.
\end{equation}
For $x>0$ we have the inequality $$0<1-\frac{1-e^{-x}}{x}<x.$$ It follows from the inequality and Lemma \ref{greedy3d}
\begin{equation}
\begin{split}
\sum_{j=0}^\infty d_j e^{-\beta\eta_j^\varphi}\left(1-\frac{1-e^{-\beta d_j}}{\beta d_j}\right)
<&\beta\sum_{j=0}^\infty d_j^2 e^{-\beta \eta_j^\varphi}\\
\ll& \beta\sum_{j=1}^\infty\eta_j^{1/4}e^{-\beta\eta_j^\varphi}+\beta e^{-\beta\eta_0^\varphi}\\
<&\beta\sum_{j=0}^\infty\eta_{j+1}^{1/4}e^{-\beta\eta_j}+\beta e^{-\beta\eta_0^\varphi}
\end{split}
\end{equation}
and the bound $N_\varphi(x)\ll x^{3/2}$ permits us to estimate the sum by the following integral:
\begin{equation}
\beta\sum_{j=0}^\infty\eta_{j+1}^{1/4}e^{-\beta\eta_j}\ll \beta\int_{0}^\infty x^{1/4}e^{-\beta x}x^{1/2}dx
\ll\beta^{-3/4}.
\end{equation}
\end{proof}

\end{document}